\newif\ifarxiv
\arxivtrue

\ifarxiv
\documentclass[9pt,journal,twoside,final]{IEEEtran}
\usepackage[numbers]{natbib}
\DeclareRobustCommand{\textcite}[1]{\citet{#1}}
\else
\documentclass[9pt,shortpaper,twoside,web]{ieeecolor}
\usepackage{generic}
\usepackage[isbn=false,
   backend=bibtex,
   style=ieee,
   bibstyle=ieee,  
   doi=false,
   sortcites=true,   
   giveninits=true,            
   backref=false]{biblatex}
\DeclareRobustCommand{\citet}[1]{\textcite{#1}}
\fi

\input{arxiv_preamble}
\usepackage{mathtools}

\usepackage[utf8]{inputenc}
\usepackage[T1]{fontenc}

\usepackage{subcaption}
\usepackage{etoolbox}

\usepackage{hyperref}
\usepackage{amssymb}
\usepackage{dsfont}  
\usepackage[nolist,nohyperlinks]{acronym}
\let\labelindent\relax
\usepackage[inline]{enumitem}  
\usepackage{csquotes} 
\usepackage{graphicx} 
\usepackage{siunitx}
\usepackage{stackengine}
\usepackage{booktabs} 

\usepackage{tikz}
\usetikzlibrary{arrows.meta}  
\usetikzlibrary{shapes.misc}  

\colorlet{lcolor}{blue!40!black}
\colorlet{ucolor}{blue!50!cyan!50!black}
\colorlet{ccolor}{blue!50!cyan!50!black}
\colorlet{fcolor}{red!40!black}
\usepackage{hypcap}
\hypersetup{colorlinks=true,linkcolor=lcolor,urlcolor=ucolor,citecolor=ccolor,filecolor=fcolor}
\newcommand{\mailto}[1]{{\href{mailto:#1}{#1}}}

\DeclareMathOperator{\Span}{span}

\renewcommand{\H}{\mathcal{H}}

\newcommand{\N}{\mathbb{N}}

\newcommand{\R}{\mathbb{R}}

\renewcommand{\S}{\mathbb{S}}

\newcommand{\X}{\mathbb{X}}

\newcommand{\Y}{\mathbb{Y}}

\newcommand{\borel}[1]{\ensuremath{\mathcal{B}(#1)}}
\newcommand{\sample}{\ensuremath{\Omega}}
\newcommand{\sigalgebra}{\ensuremath{\mathcal{A}}}

\newcommand{\kme}[1]{\ensuremath{\mu_{#1}}}
\newcommand{\dt}{\ensuremath{\Delta t}}
\newcommand{\flow}{\ensuremath{\phi}}
\newcommand{\Flow}{\ensuremath{\Phi}}
\newcommand{\outtraj}{\ensuremath{\gamma}}
\newcommand{\Outtraj}{\ensuremath{\Gamma}}
\newcommand{\lawoutput}[2]{\ensuremath{\P_{#1}^{#2}}}
\newcommand{\corrmap}{\ensuremath{K}}
\newcommand{\Yhat}{\ensuremath{\hat Y}}

\newcommand{\Gammahat}{\ensuremath{\hat \Gamma}}
\newcommand{\gammahat}{\ensuremath{\hat \gamma}}

\newcommand{\neta}{\ensuremath{{n_\eta}}}
\newcommand{\nepsilon}{\ensuremath{{n_\epsilon}}}
\newcommand{\nx}{\ensuremath{{n_x}}}
\newcommand{\ny}{\ensuremath{{n_y}}}

\newcommand{\procnoise}{\eta}
\newcommand{\measnoise}{\epsilon}
\newcommand{\normal}[2]{\mathcal{N}\left(#1, #2\right)}

\renewcommand{\d}{\ensuremath{\mathrm{d}}}

\renewcommand{\a}{\ensuremath{\mathrm{a}}}
\renewcommand{\b}{\ensuremath{\mathrm{b}}}

\newcommand{\probmeasures}[1]{\ensuremath{\mathcal{M}^+_1(#1)}}
\newcommand{\norm}[1]{{\lVert #1 \rVert}}

\newcommand{\mmd}{\ensuremath{\mathrm{MMD}}}
\newcommand{\mmdsq}{\ensuremath{\mmd^2}}

\newcommand{\mmdb}{\ensuremath{\widehat{\mmd}}}
\newcommand{\Dirac}[1]{\ensuremath{\delta_{#1}}}
\newcommand{\E}{\mathbb{E}}
\renewcommand{\P}{\mathbb{P}}
\newcommand{\Q}{\mathbb{Q}}
\newcommand{\irchi}[2]{\raisebox{\depth}{$#1\chi$}} 
\DeclareRobustCommand{\rchi}{{\mathpalette\irchi\relax}}
\newcommand{\characteristic}[1]{\ensuremath{\rchi_{#1}}}
\newcommand{\thresh}{\kappa}

\newtheorem{definition}{Definition}

\newtheorem{theorem}{Theorem}
\newtheorem{corollary}{Corollary}
\newtheorem{proposition}{Proposition}

\newtheorem{remark}{Remark}
\newtheorem{assumption}{Assumption}
\newtheorem{example}{Example}

\newacro{RKHS}[RKHS]{reproducing kernel Hilbert space}
\newacroindefinite{RKHS}{an}{a}
\newacroplural{RKHS}[RKHSs]{reproducing kernel Hilbert spaces}

\newacro{PBN}[PBN]{probabilistic Boolean network}
\newacroplural{PBN}[PBNs]{probabilistic Boolean networks}

\newacro{KME}[KME]{kernel mean embedding}
\newacroindefinite{KME}{a}{a}
\newacroplural{KME}[KMEs]{kernel mean embeddings}

\newacro{EKF}[EKF]{extended Kalman filter}
\newacroindefinite{EKF}{an}{an}

\newacro{KKL}[KKL]{Kazantzis--Kravaris/Luenberger}
\newacroindefinite{KKL}{a}{a}

\newacro{MMD}[MMD]{maximum mean discrepancy}
\newacroindefinite{MMD}{an}{a}
\newacroplural{MMD}[MMDs]{maximum mean discrepancies}

\newacro{as}[a.s.]{almost surely}
\newacro{wrt}[w.r.t.]{with respect to}

\newacro{iff}[iff]{if, and only if,}

\makeatletter
\newcommand*{\transpose}{%
  {\mathpalette\@transpose{}}%
}
\newcommand*{\@transpose}[2]{%
    \small \raisebox{1.37\depth}{$\m@th#1\intercal$}%
}
\makeatother

\newcommand{\eg}{e\/.\/g\/.,\/~}

\renewcommand{\d}{\ensuremath{\mathrm{d}}}

\usepackage[normalem]{ulem}  
\usepackage{marginnote}

\newcommand{\reworded}[1]{{\color{Magenta}#1}}
\newcommand{\added}[1]{{\color{blue}#1}}
\newcommand{\removed}[1]{{\color{red}\sout{#1}}}
\newcommand{\reviewcomment}[1]{\marginnote{\#\,#1}}

\renewcommand{\reworded}[1]{#1}
\renewcommand{\added}[1]{#1}
\renewcommand{\removed}[1]{}
\renewcommand{\reviewcomment}[1]{}

\begin{document}

\title{
 	Data-Driven Observability Analysis for Nonlinear Stochastic Systems
 }

\author{
    Pierre-François Massiani$^\dagger$\,\orcid{0000-0002-8019-4401}, Mona Buisson-Fenet$^\dagger$\,\orcid{0000-0002-7322-2101}, Friedrich Solowjow\,\orcid{0000-0003-2623-5652}, Florent Di Meglio\,\orcid{0000-0002-0232-6130}, and Sebastian Trimpe\,\orcid{0000-0002-2785-2487}
    \thanks{$^\dagger$ The authors contributed equally.}
    \ifarxiv\else\thanks{Paper initially submitted on Feb. 27\textsuperscript{th}, 2023.}\fi
    \thanks{P.-F. Massiani (\mailto{massiani@dsme.rwth-aachen.de}), M. Buisson-Fenet (\mailto{mona.buissonfenet@me.com}), F. Solowjow (\mailto{solowjow@dsme.rwth-aachen.de}), and S. Trimpe (\mailto{trimpe@dsme.rwth-aachen.de}) are with the Institute for Data Science in Mechanical Engineering, RWTH Aachen University, 52068 Aachen, Germany.}
    \thanks{M. Buisson-Fenet and F. Di Meglio (\mailto{florent.di\_meglio@minesparis.psl.eu}) are with the Centre Automatique et Systèmes, Mines Paris -- PSL university, 75006 Paris, France.}
    \thanks{M. Buisson-Fenet is also with the Ansys Research Team, Ansys France, 69100 Villeurbanne, France.}
}
\ifarxiv
\markboth{Data-Driven Observability Analysis for Nonlinear Stochastic Systems}{Data-Driven Observability Analysis for Nonlinear Stochastic Systems}
\copyrightnotice
\fi
\maketitle

\begin{abstract}
    Distinguishability and, by extension, observability are key properties of dynamical systems.
    Establishing these properties is challenging, especially when no analytical model is available and they are to be inferred directly from measurement data.
    The presence of noise further complicates this analysis, as standard notions of distinguishability are tailored to deterministic systems.
    We build on distributional distinguishability, which extends the deterministic notion by comparing  distributions of outputs of stochastic systems.
    We first show that both concepts are equivalent for a class of systems that includes linear systems.
    We then present a method to assess and quantify distributional distinguishability from output data.
    Specifically, our quantification measures how much data is required to tell apart two initial states, inducing a continuous spectrum of distinguishability.
    We propose a \emph{statistical test} to determine a threshold above which two states can be considered distinguishable with high confidence.
    We illustrate these tools by computing distinguishability maps over the state space in simulation, then leverage the test to compare sensor configurations on hardware.
\end{abstract}


\section{Introduction}

\IEEEPARstart{D}{istinguishability} is the property that allows telling apart different initial states from output measurements; observability then refers to all states being distinguishable.
It is a core assumption in observer design, and there exist multiple sufficient criteria to establish observability from analytical models\,\cite{Pauline_survey_observers}.
However, in many cases, such an analytical model is either unavailable or impractical to work with.
Examples involve engineering systems with intractable models, black-box simulators, or systems that are only partially or imperfectly modeled such that the derived properties may not transfer to the real world\,\cite{Sjoberg_nonlin_black-box_models_sysid_overview}.
In contrast, we often have access to \emph{measurement data} generated by these systems.
In this paper, we propose to infer \reworded{\reviewcomment{1.1}distinguishability} of the underlying system from such data, without relying on an explicit model of the dynamics.
\par
In many cases, measurements are corrupted by noise.
Yet, the standard definition of observability is tailored to deterministic systems\,\cite{Son1998}, and noise is often treated as a perturbation against which an observer should be robust\,\cite{Simon_optimal_state_estimation}.
\reworded{\reviewcomment{1.2}This perspective ties observability to a specific choice of observer; in contrast, we wish to define stochastic observability as an inherent property of the stochastic system, independent of the observer.}
We argue that there is a broad class of reasonable definitions, each corresponding to a way of comparing the underlying stochastic processes\,\cite{LZZC2021,FV2020,ZGG2019a,HBT2003,ZC2004,SL2010}.
We generalize \emph{distributional distinguishability}, first defined for \acp{PBN}\,\cite{LZZC2021}, to arbitrary nonlinear systems.
It consists in comparing the distributions of the output.
We show that distributional distinguishability is a reasonable generalization of the deterministic notion by proving that the stochastic and nominal systems share the same classes of indistinguishability under some assumptions, which include linear systems.
\par
\ifarxiv
\IEEEpubidadjcol  
\fi
Our main contribution is then a method to assess and quantify this distributional distinguishability from measurement data.
To achieve this, we propose to use a metric between the distributions of output measurements; its value reflects how different these distributions are.
We choose the \ac{MMD}~\cite{MFSS2017} as this metric, for three reasons.
First, it takes values in a continuous set, which induces a continuous spectrum of \emph{relative} distinguishability on which some pairs are more distinguishable than others.
Second, it can be estimated from data, hereby establishing distinguishability of initial distributions with high confidence given the data (\emph{absolute} distinguishability) through statistical testing.
Third, its continuous value is interpreted as how easily two initial distributions can be told apart, in the sense of how much data the test requires.
\par

We validate the proposed tools in simulation and on hardware.
We first demonstrate using the finite-sample \ac{MMD} to evaluate and quantify distinguishability. 
We experimentally recover the nominal system's analytical classes of indistinguishability in the cases covered by our theoretical results, and empirically demonstrate how the metric increases as we get farther from a state's class of indistinguishability.
This also enables evaluating the influence of the noise on said classes.
Further, we illustrate the difference between our test and \emph{empirical Gramians}~\cite{PM2020a}, a tool for data-driven weak observability analysis.
We then use the proposed test on a Furuta pendulum for experiment design by checking whether removing specific sensors harms observability.
Finally, we begin bridging the gap between the proposed tools and observer design by showcasing how an observer outputs the same distribution over estimated trajectories when the system is initialized in indistinguishable states.
This is a first step towards using our tools for a priori error analysis in observer design; the \ac{MMD} indicates which states observers will fail to distinguish.

\subsection{Related Work}
\label{subsec:related work}

We propose a method to assess and quantify \reworded{\reviewcomment{1.1}distinguishability} from output data without explicitly relying on a dynamics model.
To the best of our knowledge, such a model-free, data-driven analysis has never been proposed for nonlinear stochastic systems.
We start by reviewing existing notions of stochastic observability and distinguishability and corresponding metrics.
We then present an essential concept for the metric we propose:~\acp{KME}.


\subsubsection{Stochastic Observability}
Noise is often treated as a perturbation against which observers should be robust\,\cite{Simon_optimal_state_estimation}; stochastic observability then reduces to (approximate) convergence of usual observers, possibly up to some probability\,\cite{SAS1975}.
Despite its practical appeal, this approach ties observability (a system property) to the observer (a design choice), and neglects the question of how noise affects observability independently of the chosen observer. 
Answering it requires defining observability for a stochastic system; yet, there is no consensus over such a definition\,\cite{PM2020a}.
\par
Many approaches deal with linear systems and rely on the superposition principle to define observability of the whole system directly~\cite{DM2004}. 
Extensions to nonlinear systems typically result in local notions of observability.
A relevant example is stochastic Gramians\,\cite{PM2020a}, a generalization of empirical Gramians of nonlinear systems\,\cite{PM2015,Him2018}.
Similarly to our method, they enable a data-driven analysis of observability with the important difference that they build on \emph{local weak observability} instead of distinguishability; they only provide local information, as we illustrate in Section~\ref{sec:numerical}.
\par

Competing approaches aim at generalizing  \emph{distinguishability}\,\cite{Son1998} of (distributions over) initial states, rather than observability.
Core to these approaches is a way to compare the underlying \emph{stochastic processes} induced by the stochastic initial state, dynamics, and measurement.
Since stochastic processes can be compared in many more ways than deterministic trajectories, there exist different notions of distinguishability for stochastic systems.
Examples involve \acp{PBN}\,\cite{LZZC2021,FV2020,ZGG2019a} and other linear systems\,\cite{HBT2003,ZC2004,SL2010}; their corresponding definitions of distinguishability are not always explicit, but can be straightforwardly adapted from those of observability.
Interestingly, the relationship between such definitions and the existence of a probabilistic observer has not been studied, unlike for deterministic systems\,\cite{Pauline_survey_observers}.
Our approach builds on \emph{distributional distinguishability}\,\cite{LZZC2021}, which we show is equivalent to observability of the nominal system in certain cases, such as linear dynamics or additive measurement noise.
One advantage of this notion is also its amenability to statistical tests provided independent realizations of the system.

\subsubsection{Metricizing Observability}
A practical way to check observability is to summarize it in a scalar value or function.
For linear systems, such a metric is given by the rank of the observability matrix or, equivalently, of the observability Gramian\,\cite{Che1999}.
An alternative is the largest singular value of the Gramian, introducing a continuous scale and enabling a \emph{quantification} of observability.
These metrics translate locally to nonlinear systems with empirical Gramians\,\cite{PM2015};
they inform on the weak observability of each state.
Gramians and their corresponding metrics have also been extended for stochastic systems\,\cite{PM2020a}.
They are formally defined for linear systems with additive or multiplicative noise, and extensions towards arbitrary nonlinear systems are only hinted at\,\cite[Section\,5]{PM2020a}.
While our theoretical analysis shares similar restrictions, our approach differs in that \begin{enumerate*}\item our definition and quantification are readily applicable to both linear and nonlinear systems; and \item we quantify \emph{distinguishability} of two initial states, instead of weak observability.\end{enumerate*}

Finally, there is a body of research that defines observability\,\cite{LB2011} and observability metrics\,\cite{SD2014} based on information-theoretic properties of the output.
This definition strongly differs from our approach, which is thus of independent interest.

\subsubsection{Kernel Mean Embeddings}
Our quantification of \reworded{\reviewcomment{1.1}distinguishability} relies on \acp{KME} of probability distributions\,\cite{MFSS2017}.
In short, these embeddings represent probability distributions in a high-dimensional \ac{RKHS} where scalar products can be performed efficiently and often estimated from data.
Therefrom developed many methods to model and compare random variables, such as modeling of conditional relations\,\cite{SFG2013} or Bayesian inference\,\cite{FSG2013}.
Essential for this work is the \ac{MMD}, a distance between probability distributions that can be estimated from data~\cite{BGR2006}.
The \ac{MMD} was leveraged in a distribution-free two-sample test able to compare probability distributions from independent samples\,\cite{GBR2012}. 
We heavily rely on both the \ac{MMD} and this test, which requires independent samples from the same distributions. This can be problematic for dynamical systems, where the future greatly depends on the past. Some approaches alleviate this problem with mixing\,\cite{SBF2022};
in contrast, we assume independent reinitializations of the system.

\subsection{Notations}
The Borel $\sigma$-algebra of an open set $\S\subset\R^n$ is $\borel{\S}$, and $\probmeasures{\S}$ is the set of probability measures on $(\S, \borel{\S})$.
We assume open sets $\X\subset\R^\nx$ and $\Y\subset\R^\ny$, respectively called state and output spaces, and equip them with their Borel $\sigma$-algebras.
Random variables are defined on a probability space $(\sample, \sigalgebra, \P)$.
\reworded{\reviewcomment{1.4}
When defining a random variable $Z$, we abuse notations and write $Z\in\S$ to indicate concisely that $Z$ takes values in $(\S, \borel{\S})$; its characteristic function is $\characteristic{Z}$.
}
\section{Distinguishability}


We first recall a definition of distinguishability for deterministic systems (Section\,\ref{subsec:deterministic observability}), then extend this notion to nonlinear stochastic systems by generalizing distributional distinguishability of \acp{PBN} (Section\,\ref{subsec:distributional observability}).
This sets the stage for our first contribution: showing that, for linear dynamics or certain nonlinear systems, the two notions are equivalent (Section\,\ref{subsec:noise and distinguishability}).
\par
For simplicity, we focus on autonomous systems.
All concepts and results translate to control systems by considering distinguishability for a \emph{specific} controller, as is commonly done~\cite{Son1998,PM2015}.

\subsection{Deterministic Distinguishability}\label{subsec:deterministic observability}
Consider a deterministic system \begin{equation}
    \begin{split}
        x_{t+1} &= f(x_t),\\
        y_t &= h(x_t),
        \label{eq:deterministic system}
    \end{split}
\end{equation}
with arbitrary initial condition $x_0\in\X$. The function $f:\X\to\X$ is the dynamics, $h:\X\to\Y$ is the output map.
We denote by $\flow(t, x_0)\in\X$ the state of~\eqref{eq:deterministic system} initialized in $x_0$, at time $t$.
Additionally, define the \emph{output trajectory} up to time $T\in\N$ of~\eqref{eq:deterministic system} initialized in $x_0$ as \begin{equation}
    \outtraj(T, x_0) = \big(h(\flow(0, x_0)),\,\dots,\,h(\flow(T, x_0))\big)\in\Y^{T+1}.
\end{equation}
We can then define distinguishability and observability as:
\begin{definition}[{\ifarxiv \citet[Chapter\,6]{Son1998}\else \textcite[Chapter\,6]{Son1998}\fi}] \label{def:deterministic observability}
    Let $T\in\N$, and let $x_\a, x_\b\in\X$ be two initial states of~\eqref{eq:deterministic system}.
    We say that $x_\a$ and $x_\b$ are \emph{distinguishable} (in time $T$) if $\outtraj(T, x_\a)\neq\outtraj(T, x_\b)$.
    Otherwise, we say that they are \emph{indistinguishable} (in time $T$).
    Finally, the system is \emph{observable} (in time $T$) if all different initial states are distinguishable.
\end{definition}
\par 
Intuitively, indistinguishable states produce the same output until time $T$.
It is thus impossible to tell them apart by looking only at the output sequence.
In contrast, distinguishability means that the outputs differ after some time.\par
\par 
Definition~\ref{def:deterministic observability} is binary, meaning that an infinitesimal difference between trajectories can make two states distinguishable.
This becomes problematic in practice, in particular for stochastic systems, where two sample paths from the \emph{same} initial point will differ \ac{as} for nontrivial noise, dynamics, and measurement.

\subsection{Distributional Distinguishability}\label{subsec:distributional observability}

We therefore need a generalized notion of distinguishability to handle stochastic systems.
They come with the additional hurdle that output measurements are now random variables.
Hence, the question of distinguishability pertains to the more general problem of how to compare stochastic processes.
There are many suitable ways to define the ``equivalence'' of stochastic processes, each leading to a specific notion of stochastic observability~\cite{ZGG2019a,LZZC2021}.
We focus on distributional distinguishability, since the similarity with the deterministic case makes it a natural extension and because, as we will show, it lends itself to data-driven evaluation.
\par
Consider a stochastic dynamical system of the form
\begin{equation}
    \begin{split}
        X_{t+1} &= F(X_t, \procnoise_t), \\
        Y_t &= H(X_t, \measnoise_t).
        \label{eq:stochastic system}
    \end{split}
\end{equation}
The state $X_t$ and output $Y_t$ are now random variables taking values in $\X$ and $\Y$, respectively.
The variable $\procnoise_t\in\R^\neta$ represents the stochastic part of the dynamics, and $\measnoise_t\in\R^\nepsilon$ the one of the measurement.
We assume that~\eqref{eq:stochastic system} is a noisy version of~\eqref{eq:deterministic system}, that is, $F(x, 0) = f(x)$ and $H(x, 0) = h(x)$.
\par
We allow stochastic initializations, meaning that there exists an initial distribution~$\mu$ defined on $\X$ such that $X_0\sim\mu$. 
We assume that $\procnoise_t$ and $\measnoise_t$ are independent of the past states $(X_s)_{s\leq t}$.
We do not make further assumptions on the noise; e.g., we allow non-centered or autocorrelated noise.
Finally, we denote by $\Flow(t, \mu)$ the random state of~\eqref{eq:stochastic system} initialized as per $\mu$, at time $t\in\N$, and define \begin{equation}
    \Outtraj(T, \mu) = \big(H(\Flow(0, \mu), \measnoise_0),\,\dots,\,H(\Flow(T, \mu), \measnoise_T)\big)\in\Y^{T+1},
\end{equation}
the random output trajectory up to time $T\in\N$.
We denote its law by $\lawoutput{\mu}{T}$.
We are now ready to generalize distributional distinguishability, first introduced for \acp{PBN}~\cite{LZZC2021}, as follows:
\begin{definition}[Distributional Distinguishability] \label{def:distributional observability}
    Let $T\in\N$, and let $\mu_\a, \mu_\b$ be two initial distributions of~\eqref{eq:stochastic system}.
    We say that $\mu_\a$ and $\mu_\b$ are \emph{distributionally distinguishable} (in time $T$) if $\lawoutput{\mu_\a}{T}\neq\lawoutput{\mu_\b}{T}$.
    Otherwise, we say that they are \emph{distributionally indistinguishable} (in time $T$).
    Finally, the system is \emph{distributionally observable} (in time $T$) if all different initial distributions are distributionally distinguishable.
\end{definition}
\par 
Distributional observability thus compares the laws of the random trajectories instead of comparing the trajectories themselves, as Definition~\ref{def:deterministic observability} does.
Intuitively, it focuses on whether one can tell apart different initial distributions by \emph{repeating the experiment}.
\par

Since distributional distinguishability is defined for arbitrary initial distributions, it is also meaningful to talk about distributional distinguishability of \emph{initial states}; the initial distributions involved are then accordingly-centered Diracs.
In what follows, we abuse notations and use $x$ directly instead of $\delta_x$ in $\Flow(t, \cdot), \Outtraj(T, \cdot)$, and $\lawoutput{\cdot}{T}$.

\subsection{Noise and Distinguishability}
\label{subsec:noise and distinguishability}

Distributional observability reduces to the classical, deterministic notion when the system is deterministic with deterministic initializations; the distributions of $\Outtraj(T, x_\a)$ and $\Outtraj(T, x_\b)$ are then Diracs and are thus equal \ac{iff} the variables are equal \ac{as}
Therefore, we study whether noise influences distinguishability.
In particular, can introducing noise make two states distinguishable? 
Conversely, can noise hide the difference between nominally-distinguishable initial states?
In general, yes; the classes of indistinguishability of the nominal and noisy systems differ, as we illustrate in Section~\ref{sec:numerical}.
In this section, we study conditions under which the noise does \emph{not} influence distinguishability.
We start with measurement noise and show that, under a technical assumption, it does not affect distinguishability.
This implies that the deterministic and stochastic notions coincide in the absence of process noise.
\reworded{
We then treat linear systems, for which the notions always coincide even in the presence of process and measurement noise.\reviewcomment{7.3}
}

\subsubsection{Measurement Noise}

We start with the following assumption to restrict the influence of measurement noise to the output of $h$:
\begin{assumption} \label{asmptn:measurement noise output map}
    There exists a measurable function $\corrmap:\Y\times\R^\nepsilon\to\Y$ such that $H(x, \reworded{\measnoise}) = \corrmap(h(x), \reworded{\measnoise})$.
\end{assumption}
\par
This enables introducing the output before its corruption by measurement noise, a variable denoted by $\Yhat_t = h(X_t)$ independent of $\epsilon_t$.
\begin{assumption} \label{asmptn:injectivity measurement}
    For all $\nu_\a, \nu_\b\in\probmeasures{\Y^{T+1}}$ with $\nu_\a\neq\nu_\b$, there exists $A\in\borel{\Y^{T+1}}$ such that, by applying $K$ element-wise, 
    \begin{equation}\begin{split}
        \int_{\Y^{T+1}}\P[K(\gammahat, \epsilon)\in A]\d\nu_\a(\gammahat) \neq \int_{\Y^{T+1}}\P[K(\gammahat, \epsilon)\in A]\d\nu_\b(\gammahat).
    \end{split}\end{equation}
\end{assumption}
\par
This assumption is a \emph{generalized injectivity condition}: it requires the mapping from a distribution $\nu$ of $\Gammahat$ to the output distribution after corruption by measurement noise to be injective, where $\Gammahat = (\Yhat_0, \dots, \Yhat_T)$ is the output trajectory \emph{before} the measurement noise.
It ensures that if two distributions of $\Gammahat$ are distinct, they remain different after applying measurement noise.
\begin{example}[Additive Noise] \label{ex:additive noise}
Albeit technical, Assumption~\ref{asmptn:injectivity measurement} is satisfied for additive noise.
Indeed, assume $K(\gammahat, e) = \gammahat + e$, and take $\nu_\a, \nu_\b\in\probmeasures{\Y^{T+1}}$, distinct.
Let $\Gammahat_\a\sim\nu_\a$, $\Gammahat_\b\sim\nu_\b$, and assume that $\Gammahat_\a, \Gammahat_\b$, and $\epsilon$ are independent.
We have, for all $A\in\borel{\Y^{T+1}}$:\begin{equation}
    \int_{\Y^{T+1}}\P[K(\gammahat, \epsilon)\in A]\d\nu_i(\gammahat) = \P[K(\Gammahat_i, \epsilon)\in A],
\end{equation}
where $i\in\{\a, \b\}$.
Therefore, Assumption~\ref{asmptn:injectivity measurement} is satisfied \ac{iff} $\Gammahat_\a+\epsilon$ and $\Gammahat_\b+\epsilon$ have different laws, \ac{iff} $\characteristic{\Gammahat_\a+\epsilon} \neq \characteristic{\Gammahat_\b+\epsilon}.$
By independence, $\characteristic{\Gammahat_i + \epsilon} =\characteristic{\Gammahat_i}\cdot\characteristic{\epsilon}$.
Therefore, Assumption~\ref{asmptn:injectivity measurement} holds \ac{iff} $\characteristic{\Gammahat_\a} \neq \characteristic{\Gammahat_\b}$, which is true since $\nu_\a\neq\nu_\b$.
\end{example}
\par

Importantly, Assumption~\ref{asmptn:injectivity measurement} enables decoupling the effect of measurement and process noise on distributional distinguishability; measurement noise then does not affect distinguishability.
\begin{theorem} \label{thm:measurement noise}
    Let $\mu_\a, \mu_\b\in\probmeasures{\X}$ be two initial distributions, and let Assumption~\ref{asmptn:measurement noise output map} hold.
    If $\mu_\a$ and $\mu_\b$ are distributionally distinguishable, the distributions of $\Gammahat_\a\,\dot=\,(h(\Flow(t, \mu_\a))_{t=0}^T$ and $\Gammahat_\b\,\dot=\,(h(\Flow(t, \mu_\b))_{t=0}^T$ differ.
    Under Assumption~\ref{asmptn:injectivity measurement}, the converse is also true.
\end{theorem}
\begin{proof}
    Let $\mu_\a, \mu_\b\in\probmeasures{\X}$ be distinguishable, and take $A\in\borel{\Y^{T+1}}$ such that $\lawoutput{\mu_\a}{T}[A]\neq\lawoutput{\mu_\b}{T}[A]$.
    Let $i\in\{\a, \b\}$, and $\nu_i$ be the distribution of $\Gammahat_i$.
    Marginalizing $\Gamma(T, \mu_i)$ on $\Gammahat_i$ yields \begin{equation}\label{eq:marginalized lawoutput}
        \lawoutput{\mu_i}{T}[A] = \int_{\Y^{T+1}}\pi(\gammahat, A)\d\nu_i(\gammahat),
    \end{equation}
    where $\pi(\gammahat, A) = \P\left[K(\gammahat, \epsilon)\in A\right]$ is the probability of the observation lying in $A$ given that the nominal observation is $\gammahat$; it is thus independent of $\mu_\a$ and $\mu_\b$.
    Since $\lawoutput{\mu_\a}{T}[A]\neq\lawoutput{\mu_\b}{T}[A]$, \eqref{eq:marginalized lawoutput} implies that the integral of $\pi(\cdot, A)$ differs when considered w.r.t. $\nu_\a$ and $\nu_\b$.
    Therefore, $\nu_\a\neq\nu_\b$.
    Conversely, assume that $\nu_\a$ and $\nu_\b$ differ.
    Assumption~\ref{asmptn:injectivity measurement} ensures the existence of $A\in\borel{\Y^{T+1}}$ such that the right-hand side of~\eqref{eq:marginalized lawoutput} differ for $i=\a$ and $i=\b$; this shows $\lawoutput{\mu_\a}{^T}\neq\lawoutput{\mu_\b}{T}$ and concludes the proof. 
\end{proof}
\begin{corollary}
\label{coro:no process noise}
    Let Assumptions~\ref{asmptn:measurement noise output map} and~\ref{asmptn:injectivity measurement} hold.
    Assume that there is no process noise, i.e., $\eta_t=0$ \ac{as} for all $t$.
    Let $x_\a, x_\b\in\X$. Then, $x_\a$ and $x_\b$ are distinguishable for~\eqref{eq:deterministic system} \ac{iff} they are distributionally distinguishable for~\eqref{eq:stochastic system}.
\end{corollary}
\begin{proof}
    By Theorem~\ref{thm:measurement noise}, $x_\a$ and $x_\b$ are distributionally distinguishable \ac{iff} $\Gammahat_\a$ and $\Gammahat_\b$ have different distributions.
    From $\eta=0$ \ac{as}, we have $\forall x\in\X, (h(\Flow(t, x))_{t=0}^T = \gamma(T, x)$ \ac{as}
    The result follows from Definition~\ref{def:deterministic observability}.
\end{proof}
\reworded{\reviewcomment{1.6}Interestingly}, this corollary does not generalize to non-zero process noise, as we illustrate on an example in Section\,\ref{sec:numerical}.
We expect that such a generalization requires strong joint assumptions on the noise and dynamics, in the absence of which the trajectories of the nominal and noisy systems need not correlate.

\subsubsection{Linear Systems}
For linear systems, however, the superposition principle enables such a generalization.
\begin{theorem}
\label{thm:observability linear}
    \reworded{\reviewcomment{1.7}}Assume that~\eqref{eq:stochastic system} is linear, i.e., $F(x, \reworded{\procnoise}) = Ax + Q\reworded{\procnoise}$ and $H(x, \reworded{\measnoise}) = Cx + R\reworded{\measnoise}$.
    Then, two initial states $x_\a, x_\b\in\X$ are distributionally distinguishable \ac{iff} they are distinguishable for~\eqref{eq:deterministic system}.
    Additionally,~\eqref{eq:stochastic system} is distributionally observable \ac{iff}~\eqref{eq:deterministic system} is observable.
\end{theorem}
\par
This guarantees that, for linear systems, \emph{noise preserves the classes of indistinguishability} under the single assumption that the noise is independent of the past trajectory.
Additionally, Kalman's criterion on the observability matrix enables checking distributional observability.
\reworded{\reviewcomment{7.4}%
We emphasize that the equivalence between observability of the noisy and nominal systems is \emph{not} an immediate consequence of the equivalence between distinguishability of pairs of initial states;
distributional observability requires all pairs of distinct initial distributions to be distinguishable rather than initial states.
}
This statement thus means that, for a linear system, distinguishability of all initial states ensures observability; there is no gain in considering initial distributions.
\par
\begin{proof}
    For $\mu\in\probmeasures{\X}$, the superposition principle holds: \begin{equation}
        \Gamma(T, \mu) = \Gamma(T, 0) + \gamma(T, \mu),
    \end{equation}
    where we abusively denote $\gamma(T, X_0)$ with $X_0\sim\mu$ by $\gamma(T, \mu)$.
    Crucially, $\Gamma(T, 0)$ and $\gamma(T, \mu)$ are independent, since the former is driven by noise and the latter by its random initialization.
    Hence, we have $\characteristic{\Outtraj(T, \mu)} = \characteristic{\Outtraj(T, 0)}\cdot\characteristic{\outtraj(T, \mu)}$.
    We thus have the following equivalences, for $x_\a, x_\b\in\X$: they are distributionally indistinguishable \ac{iff} $\lawoutput{x_\a}{T} = \lawoutput{x_\b}{T}$, \ac{iff} $\characteristic{\Outtraj(T, x_\a)} = \characteristic{\Outtraj(T, x_\b)}$, \ac{iff} $\characteristic{\outtraj(T, \Dirac{x_\a})}=\characteristic{\outtraj(T, \Dirac{x_\b})}$, \ac{iff} $\outtraj(T, x_\a) = \outtraj(T, x_\b)$.
    This shows that $x_\a$ and $x_\b$ are distributionally indistinguishable \ac{iff} they are deterministically indistinguishable.\par
    We now show that deterministic observability implies distributional observability, since the converse follows immediately from what precedes.
    Assume deterministic observability, and let $\mu_\a, \mu_\b\in\probmeasures{\X}$ be two initial distributions with $\lawoutput{\mu_\a}{T} = \lawoutput{\mu_\b}{T}$. 
    We show $\mu_\a = \mu_\b$.
    A similar reasoning as previously shows $\characteristic{\outtraj(T, \mu_\a)} = \characteristic{\outtraj(T, \mu_\b)}$, i.e., $\outtraj(T, \mu_\a)$ and $\outtraj(T, \mu_\b)$ have the same law and, thus, $\mu_\a$ and $\mu_\b$ coincide on all sets $A$ in the $\sigma$-algebra $\mathcal{S} = \gamma^{-1}(T, \borel{\Y^{T+1}})\subset\borel{\X}$.
    We show $\mathcal{S} = \borel{\X}$.
    From observability of~\eqref{eq:deterministic system}, $\gamma(T, \cdot)$ is injective.
    Since it is also linear, it admits a continuous (and thus, measurable) left inverse, say, $\lambda:\Y^{T+1}\to\X$.
    We are now ready to conclude.
    Take $A\in\borel{\X}$, and define $B = \lambda^{-1}(A)\in\borel{\Y^{T+1}}$.
    We have, by definition of $\lambda$ as a left inverse of $\gamma$, $\gamma^{-1}(T, B) = \gamma^{-1}(T, \lambda^{-1}(A)) = A$.
    This shows that $A\in\mathcal{S}$, i.e., that $\borel{\X} = \mathcal{S}$.
    Therefore, $\mu_\a(A) = \mu_\b(A)$ for all $A\in\borel{\X}$, which shows distributional observability of~\eqref{eq:stochastic system}.
\end{proof}

Albeit distributional and deterministic distinguishability differ in general, it is reasonable to expect that a small amount of noise only mildly affects the distinguishability of two states.
In the next section, we propose a \emph{quantification} of distinguishability to formalize the above intuition, and a method to estimate it from output data.
\section{Measuring Distributional Observability}
\label{sec:measuring observability}
Distributional \reworded{\reviewcomment{1.1}distinguishability} is defined as the equality of certain probability distributions.
We propose to extend this binary notion to a continuous spectrum of relative distinguishability by quantifying it with the \ac{MMD}, a distance between said distributions measuring how easily they can be told apart.
Additionally, the \ac{MMD} can be estimated from data.
This leads to our main methodological contribution: use output measurements to \begin{enumerate*}
    \item approximate relative distinguishability; and
    \item leverage a two-sample test~\cite{GBR2012} to assess distributional \reworded{\reviewcomment{1.1}distinguishability} with high confidence.
\end{enumerate*}

\subsection{The MMD to Measure Relative Distinguishability}
\label{sec:relative distinguishability}
We start with background necessary to define the \ac{MMD} and state the results that justify the method we propose.

\subsubsection{Background}
Before introducing the~\ac{MMD}, we need to define~\acp{RKHS} and~\acp{KME}.
We focus on a minimal treatment and refer the interested reader to a dedicated survey~\cite{MFSS2017} for an overview. 
\par
First, one picks a positive definite \emph{kernel function} $k:\Y^{T+1}\times\Y^{T+1}\to\R$.
Such a kernel function uniquely determines an underlying Hilbert space of functions $\H$ defined on $\Y^{T+1}$, called~\iac{RKHS}.
Then, one embeds a given probability distribution $\P\in\probmeasures{\Y^{T+1}}$ into this~\ac{RKHS} through its \ac{KME} $\kme{\P} = \E_{X\sim\P}[k(\cdot, X)]$.
The~\ac{MMD} between two distributions $\P$ and $\Q$ is then the distance in the~\ac{RKHS} between their embeddings:
\begin{align}
    \mmd[\P, \Q] = \norm{\kme{\P} - \kme{\Q}}_\H,
    \label{eq:mmd}
\end{align}
where $\norm{\cdot}_\H$ is the norm in the \ac{RKHS} $\H$.
In general, the~\ac{MMD} is only a semi-metric, since the embedding~$\P\mapsto\mu_\P$ may not be injective.
Kernels for which it \emph{is} injective are called \emph{characteristic}\,\cite{SFL2011}; they include the Gaussian kernel.
From now on, we assume a characteristic kernel.

\begin{remark}[A Kernel on Trajectories]
    We emphasize that the considered kernel $k$ is defined on the output trajectory space $\Y^{T+1}$, and not on the state space $\X$ nor on the output space $\Y$.
    For instance, the expression of a Gaussian kernel on $\Y^{T+1}$ with scalar width $\sigma>0$, which is known to be characteristic~\cite[Table\,3.1]{MFSS2017}, is
    \begin{equation}\label{eq:gaussian rbf}
        k(\outtraj_\a, \outtraj_\b) = \exp\left[-\frac{1}{2\sigma^2} \sum_{t=0}^T\sum_{i=1}^{\ny}(\outtraj_{\a,i,t} - \outtraj_{\b,i,t})^2\right],
    \end{equation}
    where $\outtraj_\a,\outtraj_\b\in\R^{(T+1)\times\ny}$ are output trajectories.
\end{remark}
\vspace{2mm}

\subsubsection{A Metric of Distinguishability}\label{sec:interpretation mmd spectrum}
The value of the \ac{MMD} thus indicates whether initial distributions are distinguishable.
\begin{proposition} \label{prop:mmd metric}
    Let $\mu_\a, \mu_\b$ be two initial distributions.
    Assume that the kernel $k$ is characteristic. 
    Then, $\mmd[\lawoutput{\mu_\a}{T}, \lawoutput{\mu_\b}{T}] = 0$ \ac{iff} $\mu_\a$ and $\mu_\b$ are distributionally indistinguishable.
\end{proposition}
\begin{proof}
The~\ac{MMD} is a metric when the kernel is characteristic~\cite{SFL2011}, which implies that $\mmdsq[\lawoutput{\mu_\a}{T}, \lawoutput{\mu_\b}{T}] = 0$ \ac{iff} $\lawoutput{\mu_\a}{T} = \lawoutput{\mu_\b}{T}$. 
The result follows from Definition~\ref{def:distributional observability}.
\end{proof}
\par
The information the \ac{MMD} carries is richer than a simple yes-or-no indication.
Indeed, it takes continuous values; low values reveal similar distributions.
Therefore, the value of the \ac{MMD} extends the binary notion of \emph{absolute} distinguishability to a continuous spectrum of \emph{relative} distinguishability on which some pairs of initial distributions are more (or less) distinguishable than others.
We postpone further interpreting this value to Section~\ref{sec:interpretaion mmd data}

\subsection{Finite-Sample Approximation}
To access both the relative and absolute notions, one needs to evaluate the~\ac{MMD}.
We rely on a finite-sample approximation~\cite{GBR2012} for this and use a two-sample test to conclude on (absolute) distinguishability, since the approximation is never exactly $0$.
Finally, this test helps refine the interpretation of the \ac{MMD} value.

\subsubsection{Background -- Finite-Sample MMD}
While a closed-form formula is often unattainable or requires perfect system knowledge, the~\ac{MMD} can be approximated from independent samples of output trajectories by the following empirical estimator.
\begin{theorem}[{\citet{GBR2012}}]\label{thm:estimator}
    Let $\mu_\a, \mu_\b$ be two initial distributions.
    Let $\Outtraj_{\a,1},\dots, \Outtraj_{\a,m}$ be independent copies of $\Outtraj(\mu_\a, T)$, and $\Outtraj_{\b,1}, \dots, \Outtraj_{\b,n}$ independent copies of $\Outtraj(\mu_\b, T)$.
    Assume that $k$ is bounded, i.e., $0 \leq k \leq B$ for some $B\in\R$.
    Then,
    \begin{align}
        \mmdb^2[m, n] & = \frac{1}{m^2}\sum_{i,j=1}^m k(\Outtraj_{\a,i}, \Outtraj_{\a,j}) 
             + \frac{1}{n^2}\sum_{i,j=1}^n k(\Outtraj_{\b,i}, \Outtraj_{\b,j})
             \notag
             \\
            &\qquad\quad - \frac{2}{mn} \sum_{i,j=1}^{m,n}k(\Outtraj_{\a,i}, \Outtraj_{\b,j}),
        \label{eq:mmd biased estimator}
    \end{align}
    converges in probability to $\mmd^2[\lawoutput{\mu_\a}{T}, \lawoutput{\mu_\b}{T}]$ with the following concentration bound for $\varepsilon>0$:
    \begin{align}
        \P&\left[
            \Big|\mmdb[m,n] - \mmd[\lawoutput{\mu_\a}{T}, \lawoutput{\mu_\b}{T}]\Big|
            > 
            2\left(\sqrt{\frac{B}{m}}+\sqrt{\frac{B}{n}}\right) + \varepsilon
        \right]~
        \notag
        \\
        & \hspace{3em} \leq
        2\exp\left[
            -\frac{\varepsilon^2}{2B}\cdot\frac{mn}{m+n}
        \right] .
        \label{eq:concentration bound}
    \end{align}
\end{theorem}
\par
\vspace{0.5em}
The outcome of~\eqref{eq:mmd biased estimator} is a noisy, biased, and finite-sample approximation of the population value $\mmd[\lawoutput{\mu_\a}{T},\lawoutput{\mu_\b}{T}]$.
We can rely on the value of the estimator to approximate relative distinguishability, however, we cannot use it directly to check absolute distinguishability.
Instead, we use a statistical test.

\subsubsection{A Two-Sample Test of Distinguishability}

To determine whether two initial distributions $\mu_\a$ and $\mu_\b$ are distinguishable from the value of~\eqref{eq:mmd biased estimator}, we perform a statistical test of the null hypothesis \begin{equation}
    H_0:~\lawoutput{\mu_\a}{T} = \lawoutput{\mu_\b}{T}\quad\text{vs.}\quad~H_1:~\lawoutput{\mu_\a}{T} \neq \lawoutput{\mu_\b}{T}.
\end{equation}
Such a test for equality of distributions is called a \emph{two-sample test}.
We propose to use one based on the concentration bound~\eqref{eq:concentration bound}\,\cite{GBR2012}.
One specifies an acceptable Type I error\footnote{\added{\reviewcomment{7.5}The Type I error is the probability of incorrectly rejecting $H_0$, whereas the Type II error is the probability of incorrectly failing to reject it.}} $\alpha\in(0,1)$, and the test provides an \emph{acceptance region threshold}; if the outcome of~\eqref{eq:mmd biased estimator} falls outside of that region, one rejects the null hypothesis $H_0$ with confidence level at least $1 - \alpha$.

\begin{proposition} \label{thm:test}
    Let $\alpha\in(0,1)$. Under the setting of Theorem~\ref{thm:estimator} with $n=m$, consider a realization of $\mmdb[m,m]$.
    If this empirical estimate is such that
    \begin{equation} \label{eq:two sample test}
        \mmdb[m,m] \geq \sqrt{\frac{2B}{m}}\left(1 + \sqrt{2\ln \alpha^{-1}}\right) \dot{=}\,\thresh,
    \end{equation}
    then $\mu_\a$ and $\mu_\b$ are distributionally distinguishable in time $T$ with confidence level $1 - \alpha$.
\end{proposition}
\begin{proof}
This immediately follows from~\cite[Corollary\,9]{GBR2012}.
\end{proof}
The threshold $\kappa$ is conservative for many practical applications.
It is common to use \emph{bootstrapping} instead, as suggested in~\cite{GBR2012} and implemented in their original code.
We use this bootstrapped value for $\kappa$ in Section~\ref{sec:numerical}.

\subsubsection{From Qualitative Properties to Sample Bounds} \label{sec:interpretaion mmd data}
In Section~\ref{sec:interpretation mmd spectrum}, we interpreted the \ac{MMD} as a measure of distinguishability based on the argument that it takes values in a continuous set.
The two-sample test of Proposition~\ref{thm:test} allows us to refine this interpretation: the value of the \ac{MMD} indicates how much data is required for the test to successfully reject the null hypothesis $H_0$ with high confidence.
In other words, the \ac{MMD} provides a sample bound for the test to achieve a specific Type II error\footnotemark[1] $\beta\in(0, 1)$.
\begin{corollary}
    Let $\beta\in(0, 1)$. Under the setting of Theorem~\ref{thm:estimator} with $n=m$, assume that the null hypothesis is incorrect and let $z = \mmd[\lawoutput{\mu}{T}, \lawoutput{\nu}{T}] > 0$.
    Assume that \begin{equation}
        m > \frac{B}{z^2}\left(4 + \sqrt{2} + 2\left(\sqrt{\ln\frac{1}{\beta}}+\sqrt{\ln\frac{2}{\beta}}\right)\right)^2. \label{eq:minimum data}
    \end{equation}
    Then, the null hypothesis is rejected with probability at least $1 - \beta$.
\end{corollary}
\begin{proof}
    Introduce $A_m = \{\mmdb[m,m]<\kappa\}$ the test's acceptance region of level $\beta$ and $B_m = \{\lvert\mmdb[m,m] - z\rvert \leq \delta\}$ the high-probability region of the concentration bound~\eqref{eq:concentration bound}, with \begin{equation}
        \epsilon = 2\sqrt{\frac{B}{m}\ln\frac{2}{\beta}},~\text{and}~ \delta = 2\sqrt{\frac{B}{m}}\left(2 + \sqrt{\ln\frac{2}{\beta}}\right).
    \end{equation}
    If $A_m\cap B_m = \emptyset$, then $H_0$ is rejected with probability at least $1 - \beta$; indeed, we then have $\P[A_m]\leq \P[B_m^\complement]\leq \beta$.
    It is sufficient for $A_m\cap B_m=\emptyset$ that $\kappa < z - \delta$.
    Rearranging yields~\eqref{eq:minimum data}.
\end{proof}

Albeit the threshold~\eqref{eq:minimum data} depends on $z$ and thus cannot be used in practice to guide data collection, it sheds light on the interpretation of the \ac{MMD} by translating a qualitative property (it measures how different outputs from two initial distributions are) into practical considerations (how much data is required to tell them apart).
Intuitively, initial distributions leading to a high~\ac{MMD} can be told apart with confidence given only a small number of output trajectories.
\section{Experimental Results} \label{sec:numerical}

\begin{figure}
    \centering    
    \includegraphics[width=0.7\columnwidth]{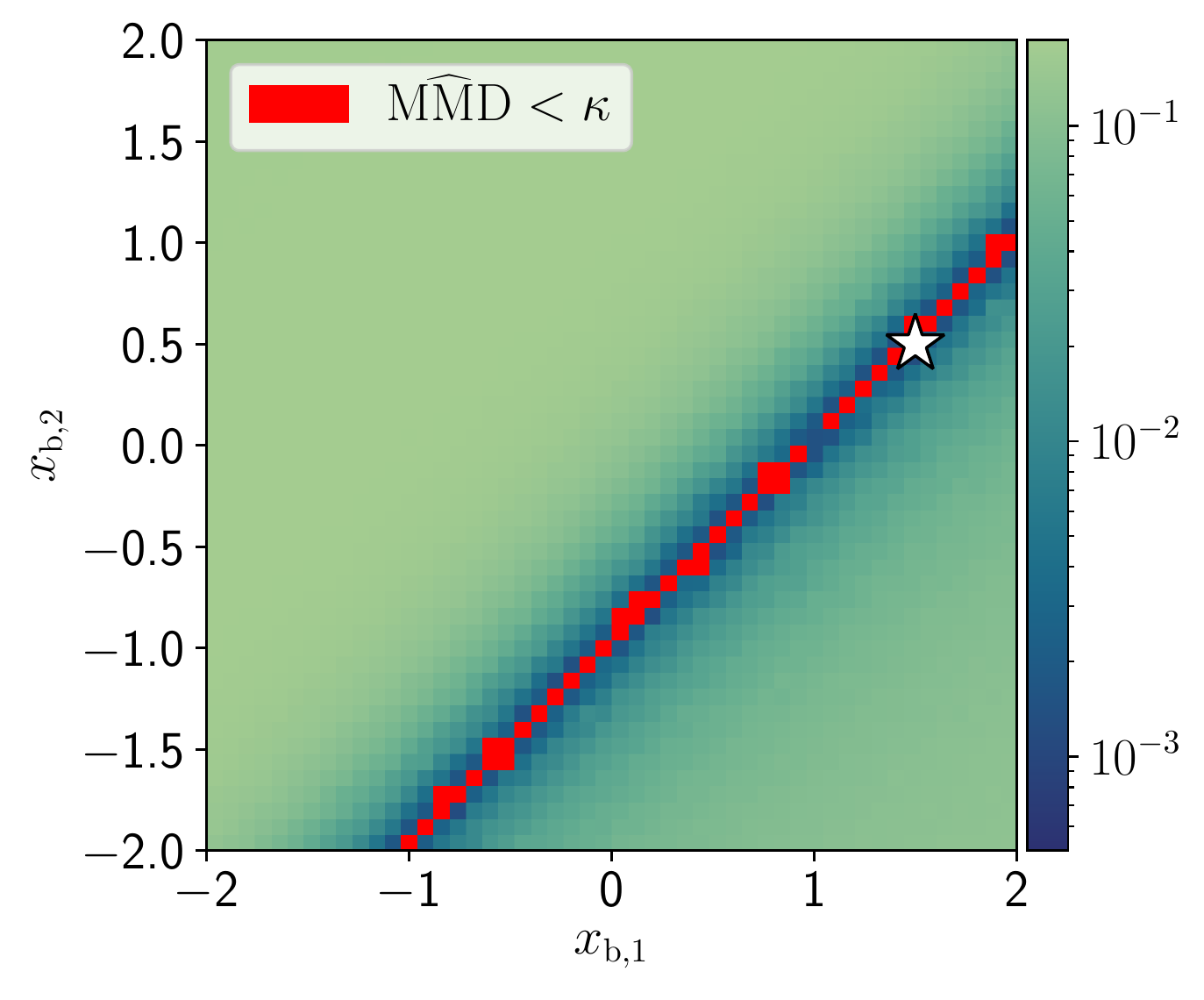}
    \caption{\ac{MMD} over an $x_\b$-grid for linear system~\eqref{def_syst_drift} from reference point $x_\a=(1.5, 0.5)$ (white star). The empirical class of indistinguishability (red points) is computed using output data from the noisy system. It matches the analytical class of the nominal system.}
    \label{fig_drift_syst_MMD}
\end{figure}

The proposed tools enable a data-driven \reworded{\reviewcomment{1.1}distinguishability} analysis of nonlinear stochastic systems.
We first verify that adding noise preserves distributional observability for linear systems, as stated in Theorem~\ref{thm:observability linear}.
We then illustrate the real-valued quantification by computing it on the whole state space for a nonlinear Duffing oscillator, revealing a continuous increase as we get away from the class of indistinguishability.
This also enables measuring the effects of both measurement and process noise on said classes, illustrating the results of Section~\ref{subsec:noise and distinguishability}.
Finally, we demonstrate on a Furuta pendulum how to test whether a sensor configuration distinguishes certain states, with implications on experiment design.
All results use the Gaussian kernel~\eqref{eq:gaussian rbf} and $\alpha=0.05$.
\added{\reviewcomment{7.7,\\10.2,10.4}
We do not report a hyperparameter study over $\alpha$ or $T$.
The former is a decision-theoretic parameter whose choice depends on downstream applications; it controls the threshold $\kappa$ and thus the test's outcome, but does not influence MMD values.
The effect of the parameter $T$ was found negligible\footnote{Code to reproduce the results and a hyperparameter study is available at~\url{https://github.com/Data-Science-in-Mechanical-Engineering/data-obs.git}.}.}
\begin{remark}
We focused on discrete-time systems; yet, dynamical systems are often described in continuous-time and only the measurements are discrete-time.
This discrepancy is not problematic as the results of Section~\ref{sec:measuring observability} hold if the state is a continuous-time process.
\end{remark}

\begin{figure*}
    \centering
    \includegraphics[width=0.72\textwidth]{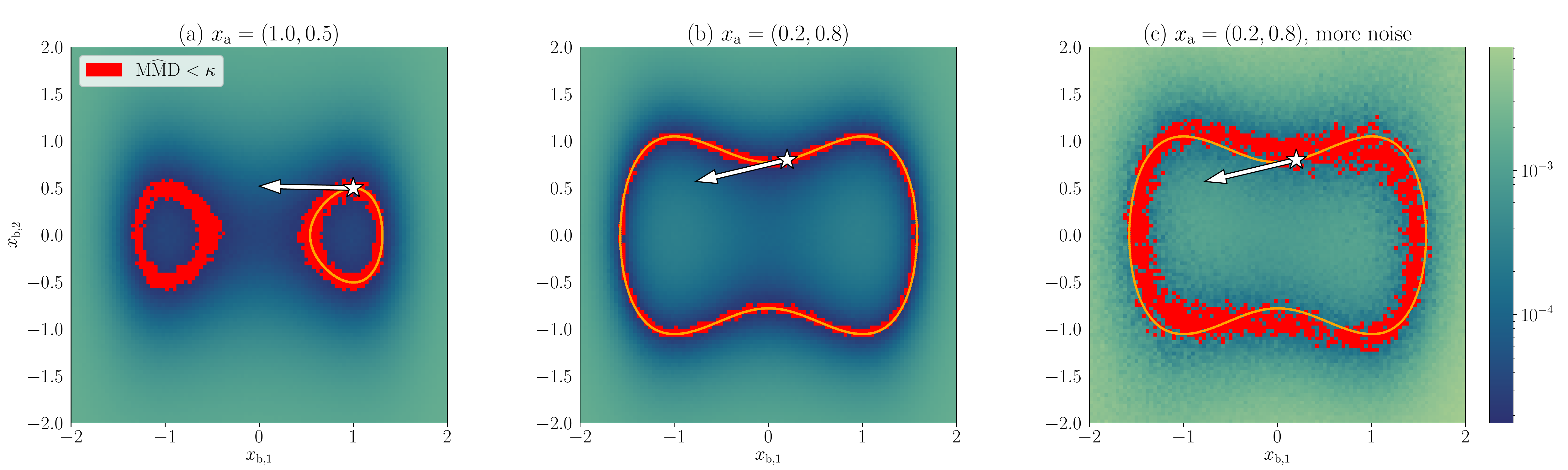}
    \caption{\ac{MMD} over an $x_\b$-grid for the Duffing oscillator~\eqref{eq:duffing} for different values of $x_\a$ (white stars). The trajectory without noise (orange) initialized in $x_\a$ is known to be a subset of the class of indistinguishability. The states where the test did not trigger (red points) constitute the empirical class of indistinguishability. In (a), this class contains not only the trajectory starting from $x_\a$, but also its symmetric w.r.t. the origin. In (c), the empirical class of distinguishability differs due to the significant process noise ($b_1=b_2=0.5$ rather than $0.05$).
    In all cases, vectors generating the null space of the nominal system's empirical Gramian (white arrow) are tangent to its class of indistinguishability.}
    \label{fig_MMD_heatmaps_duffing}
\end{figure*}
\begin{remark}
A common heuristic to select the width parameter $\sigma$ is to pick the median pairwise distance between the data sets at hand\,\cite{GBR2012}.
It is not directly applicable here, however: our interpretation of the~\ac{MMD} as a measure of relative distinguishability (Section~\ref{sec:relative distinguishability}) requires that $\sigma$ be the same for all points to be compared.
Indeed, comparing different values of~\ac{MMD} is only meaningful if the underlying embeddings are all in the \emph{same}~\ac{RKHS}, that is, use the same kernel function and parameter values.
Therefore, we suggest the following meta-heuristic. 
We start by computing the above heuristic for each pair of data sets of interest. 
Then, we set $\sigma$ as the 10-quantile of all previous values. 
These low values correspond to pairs of initial states that are indistinguishable, or hardly distinguishable, so $\sigma$ has the correct order of magnitude for the noise of the system.
We run all experiments with the fixed value of $\sigma$ given by this meta-heuristic, 
\reworded{\reviewcomment{10.8}which is computed automatically in Sections\,\ref{subsec:example linear} and~\ref{subsec:example duffing}, and reported in Table~\ref{table_qube_results} for Section~\ref{subsec:qube}.}
\removed{\reviewcomment{7.6}%
This leads to satisfying results in all our use cases, however, it is but an arbitrary choice; one should ensure $\sigma$ is set to a suitable order of magnitude or the problem might be ill-conditioned and all \ac{MMD} values close to $0$.
}
\end{remark}

\subsection{Influence of Noise on Observability} \label{subsec:example linear}

Theorem~\ref{thm:observability linear} states that two states are distributionally distinguishable for a linear stochastic system \ac{iff} they are distinguishable for the corresponding deterministic system.
This can be interpreted as follows: for linear systems, the presence of noise does not alter distinguishability between two states and, by extension, observability.
We verify this numerically on the following linear system:
\begin{equation} \label{def_syst_drift}
    \begin{split}
        \d X_t & = A X_t \d t + A_0 \sin(\omega t) \d t  + \Sigma \d W_t ,\\
        Y_t & = C  X_t + \measnoise_t,\\
        A & = \begin{pmatrix} 
            -2 & -1 \\ -1 & -2 
        \end{pmatrix}, \quad
        C = \begin{pmatrix} -1 & 1 \end{pmatrix} ,
    \end{split}
\end{equation}
where $\d W_t$ is a 2-dimensional independent Wiener process, $A_0= (3,\, 3)$, $\omega=2$, $\Sigma=0.1 I_2$, and $\measnoise_t \sim \normal{0}{0.01}$ is Gaussian measurement noise. 
The term $A_0 \sin(\omega t)\d t + \Sigma \d W_t$ is a non-centered process perturbation.
The nominal linear system $(A, C)$ is not observable.
More specifically, the direction $C^\perp = \Span\{(1, 1) \}$ cannot be observed, meaning that the set of states indistinguishable from $x\in\R^2$ is $x + C^\perp$.
This set is called the \emph{class of indistinguishability} of $x$.
We show experimentally that the non-centered noise in~\eqref{def_syst_drift} preserves this class of indistinguishability.
We emphasize that the process noise has nonzero mean over the simulation time.

We select a reference point $x_\a = (1.5, 0.5)$ and simulate $n=30$ output trajectories for $\SI{2}{\second}$ starting from $x_\a$, with time steps of length $\dt = \SI{0.01}{\second}$, yielding discrete trajectories of length $T=200$. 
These trajectories constitute a data set of samples from $\lawoutput{x_\a}{T}$. 

We then build a grid of $50 \times 50$ points and simulate $m=30$ output trajectories starting from each point $x_\b$ in the grid. We compute the \ac{MMD} between the data set from $x_\a$ and the one for each $x_\b$.
This yields a heatmap of relative distinguishability depicted in Figure~\ref{fig_drift_syst_MMD}.
We observe that the \ac{MMD} is minimal along the diagonal $x + C^\perp$, and increases as the query point gets farther. 
This shows that the farther from the diagonal a state is, the more distinguishable it is from $x_\a$.
In fact, the test of Proposition~\ref{thm:test} does not trigger along this diagonal (red dots), and triggers everywhere else, empirically identifying this set as the class of \emph{distributional} indistinguishability of $x_\a$ given the output data. 
This class indeed corresponds to the class of \emph{deterministic} indistinguishability of $x_\a$ for the unperturbed version of~\eqref{def_syst_drift} (computed analytically), as predicted by Theorem~\ref{thm:observability linear}.

\begin{figure*}
    \centering
    \begin{subfigure}{0.15\textwidth}
        \centering
        \input{img/Fig3/subfigure_a.tikz}
        \caption{Furuta pendulum}
        \label{fig:scheme furuta}
    \end{subfigure}
    \hspace{2em}
    \begin{subfigure}{0.3\textwidth}
        \centering
        \includegraphics[width=\textwidth]{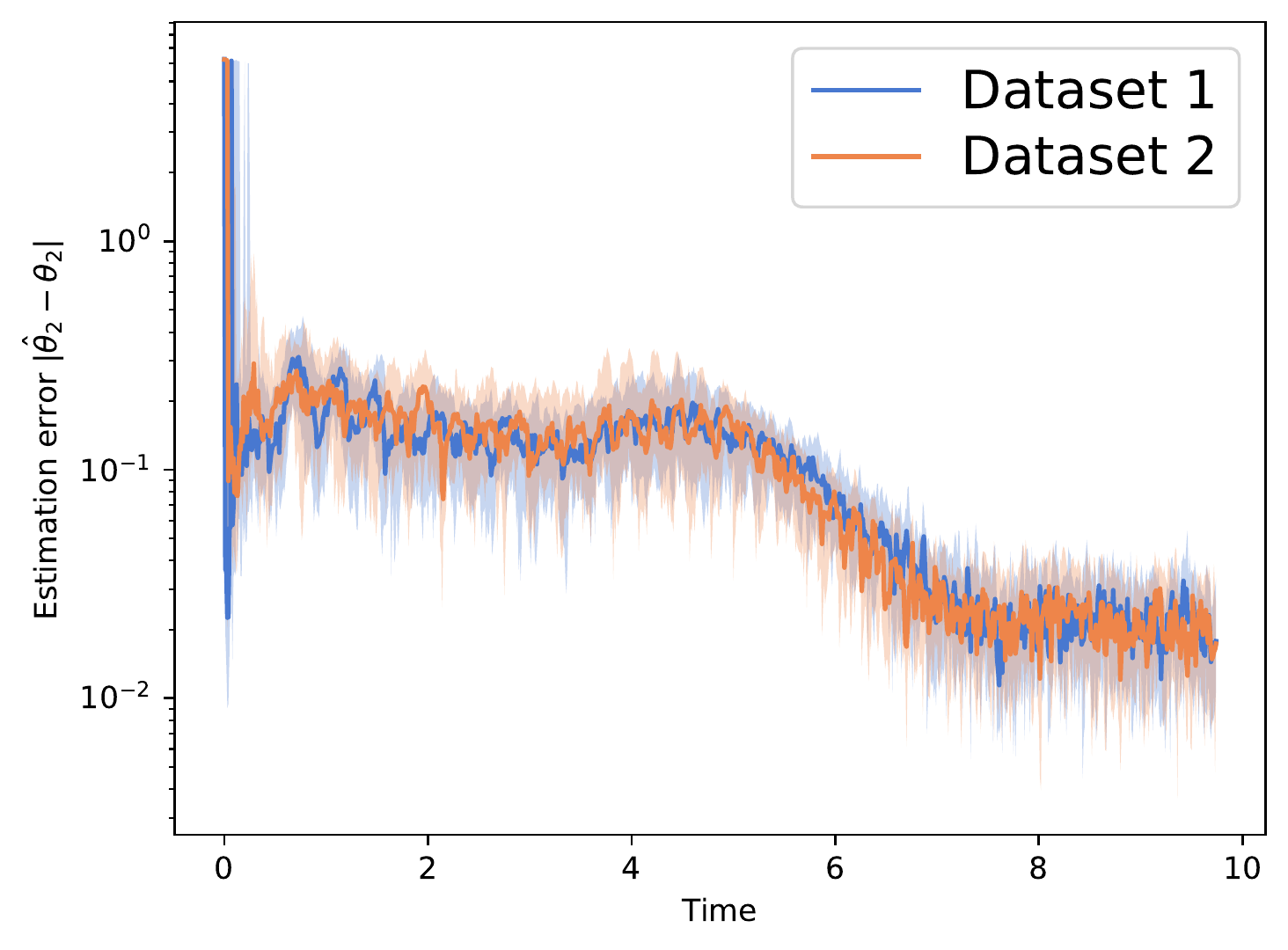}
        \includegraphics[width=\textwidth]{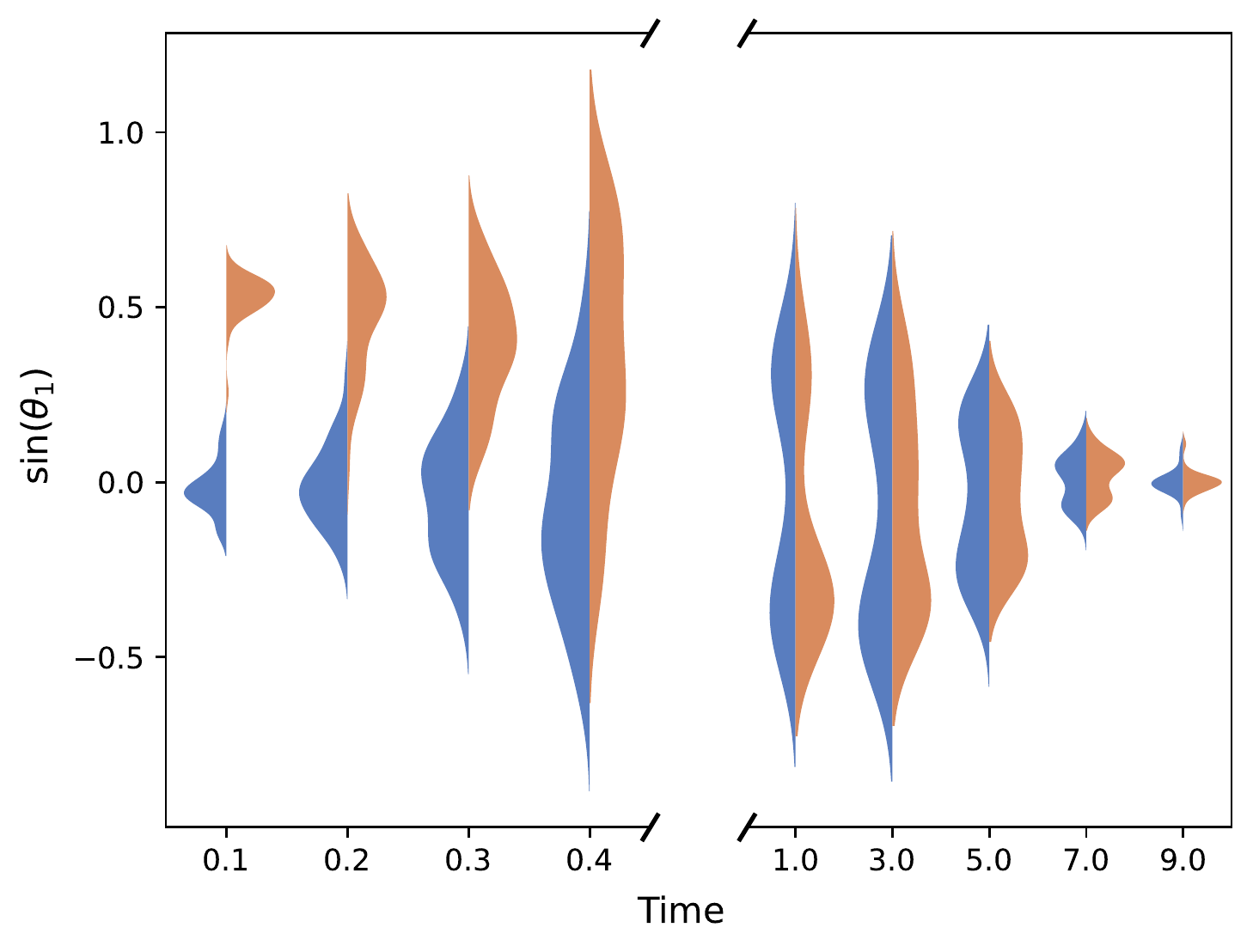}
        \caption{$y=\theta_1$}
        \label{fig:violin distinguishable}
    \end{subfigure}
    \begin{subfigure}{0.3\textwidth}
        \centering
        \includegraphics[width=\textwidth]{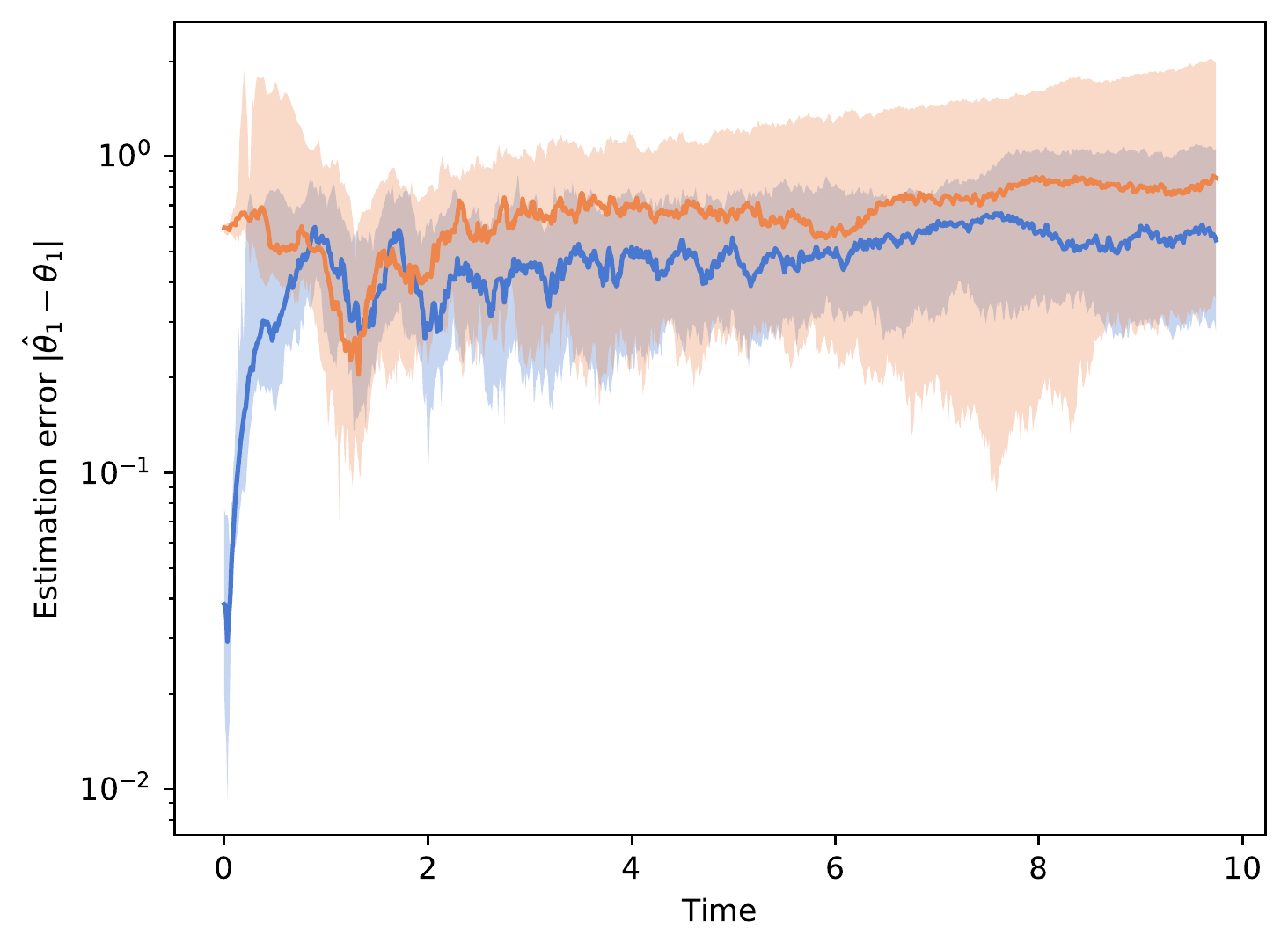}
        \includegraphics[width=\textwidth]{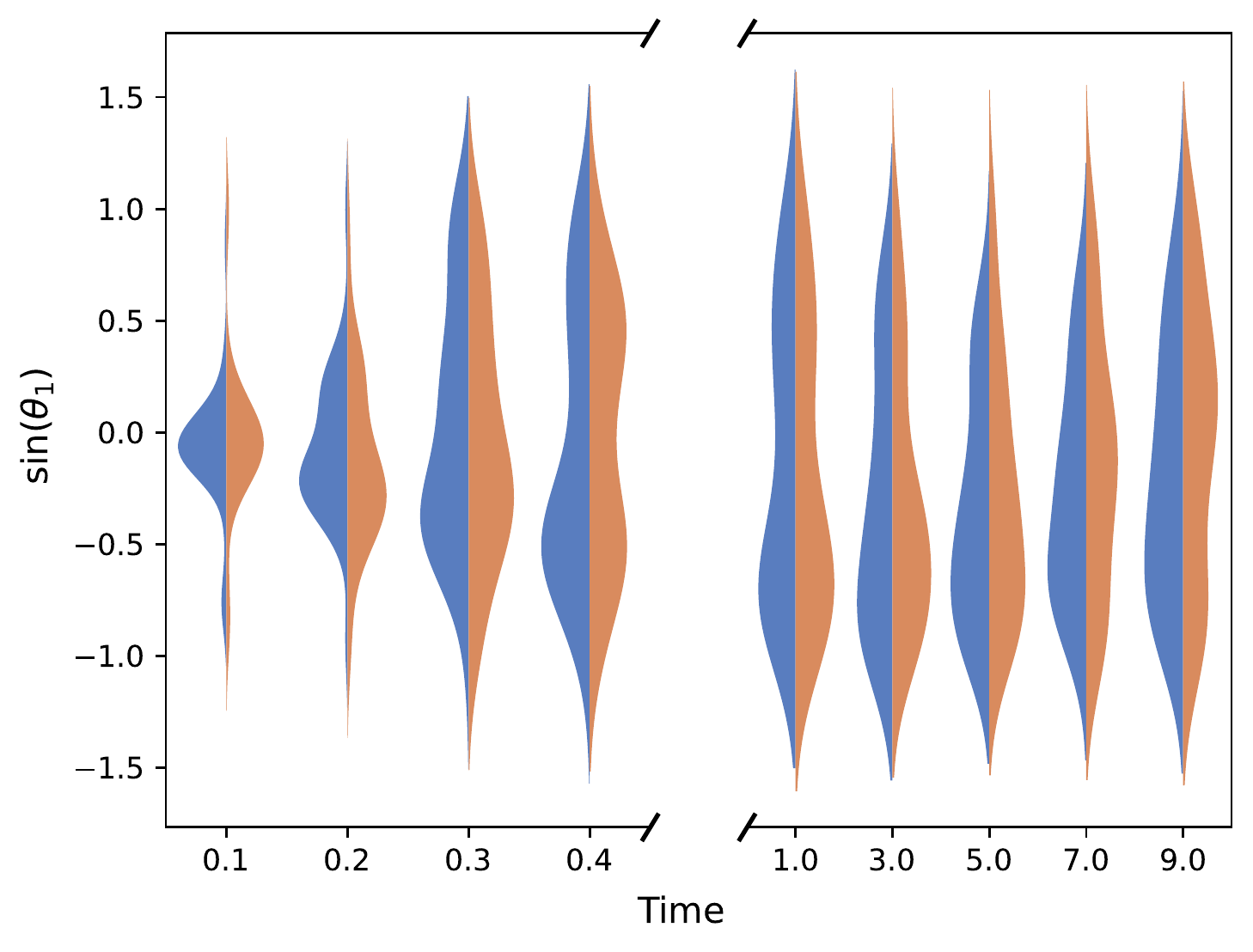}
        \caption{$y=\theta_2$}
        \label{fig:violin indistinguishable}
    \end{subfigure}
    \caption{
        Scheme of the Furuta pendulum (a), and behavior of the \ac{EKF} over time (b and c): estimation error (median and quartiles) \added{of the non-measured state} (top), and distribution of the estimated $\sin(\theta_1)$ (bottom).
        In the distinguishable case (b), the EKF converges \added{(top)} and the distributions of the predictions differ during the transient \added{(bottom)}.
        In the indistinguishable case (c), the estimation error remains high \added{(top)} and the distributions similar \added{(bottom)}; the output of the observer is statistically the same for both initial states.
        Distributions (bottom) are interpolated by kernel density estimation.
    }%
    \added{\reviewcomment{7.8}}
\label{fig_quanserqube_EKF}
\end{figure*}

\subsection{Analyzing \reworded{Distinguishability} in the State Space}\label{subsec:example duffing}

The proposed approach can be used to determine empirically the class of (distributional) indistinguishability of any state for an arbitrary nonlinear system, and to observe the continuous increase in relative distinguishability as we get farther from that class. We illustrate this with an undamped, unforced Duffing oscillator, often used in nonlinear observer design~\cite{Astolfi_uniting_observers_long,buisson2022towards}:
\begin{equation}\label{eq:duffing}
\begin{aligned}
    \d X_1 & = X_2 \d t + b_1 \d W_{1,t},
    \\
    \d X_2 & = ( X_1  -  X_1^3 ) \d t + b_2 \d W_{2,t},
\end{aligned}
\end{equation}
where $\d W_{1,t}$, $\d W_{2,t}$ are independent Wiener processes.
The nominal system is Hamiltonian and conserves the quantity
\begin{align}
h (x) & = -\frac{1}{2} x_1^2 + \frac{1}{2} x_2^2 + \frac{1}{4} x_1^4
\end{align}
along trajectories. We consider the output map $Y_t=h(X_t) + \measnoise_t$ with $\measnoise$ centered Gaussian measurement noise of variance $0.5$, and two noise configurations: $b_1=b_2=0.05$ (low process noise), and $b_1=b_2=0.5$ (high process noise).
Given an arbitrary point $x_\a$, we aim to determine its class of indistinguishability from samples of the stochastic system.
For this, we generate $n=50$ output trajectories of $\SI{1}{\second}$ initialized in $x_\a$, and $m=50$ trajectories in each other point $x_\b$ on a $100 \times 100$ grid, with $\dt=\SI{0.001}{\second}$ ($T=1000$).
We then compute the empirical \ac{MMD}~\eqref{eq:mmd biased estimator} between the resulting data sets.
\par 
Fig.~\ref{fig_MMD_heatmaps_duffing} shows the resulting \ac{MMD} maps for two values of $x_\a$, and the two levels of process noise. 
For low process noise (a and b), the \ac{MMD} increases as the discrepancy in $h$ in the initial states does; the states get more easily distinguishable.
Additionally, the empirical class of indistinguishability (red squares) recovers the nominal system's one, which indeed contains the deterministic trajectory starting from $x_\a$ (orange curve), where $h$ is constant.
For higher levels of process noise \added{(c)}, the empirical class is deformed, \added{\reviewcomment{10.3}showing empirically that process noise can alter classes of indistinguishability and hereby }illustrating that Corollary~\ref{coro:no process noise} does not generalize without further assumptions.
This continuous effect of the level of process noise on the classes highlights the interest of a continuous-valued metric of relative distinguishability; states identified as distinguishable (the test triggers) may still be \emph{difficult} to distinguish.
\par 
\begin{remark}
\label{rem:Gramians}
For comparison, we also compute the nominal system's empirical observability Gramian at $x_\a$~\cite[Eq.\,4]{PM2015}.
If the null space $N$ of this empirical Gramian is $\{0\}$ in the limit $\epsilon\to 0$, then the system is weakly observable in $x_\a$~\cite[Thm.\,2]{PM2015}.
Based on~\cite[Eq.\,(14)]{PM2015}, we thus interpret $N$ as the ``direction of weak unobservability'' in $x_\a$.
We approximate the assumption $\epsilon\to 0$ by taking $\epsilon=0.1$ and considering the eigenspace of the smaller eigenvalue.
Fig.~\ref{fig_MMD_heatmaps_duffing} shows the resulting approximation of $N$ by plotting a generating vector.
We find that the null space is the tangent space to the class of indisinguishability our test finds.
In other words, empirical Gramians provide the local directions of indistinguishability, while our test recovers the whole class.
Recovering the information on local observability of $x_\a$ could also be achieved with our test; e.g., by checking the test outcome in an $x_\b$-spherical shell centered in $x_\a$.
Indeed, states very close to $x_\a$ will most likely not trigger the test due to finite-sample approximations, and should thus be excluded.
We leave this connection to local observability for future work.
\end{remark}

\subsection{Analyzing Sensor Configurations on Hardware}\label{subsec:qube}

When studying an experimental system, one can sometimes choose between possible sensor configurations whose respective benefits can be uncertain.
The system may be observable and allow for building an observer to achieve further downstream tasks in some of these configurations but not in others.
We now demonstrate how the proposed tools can help compare these configurations from measurements.
For this, we collect hardware data on a Furuta pendulum: the Qube Servo 2 by Quanser~(Fig.~\ref{fig:scheme furuta}).
The pendulum consists of two arms, the first of which is actuated, with respective angular positions $(\theta_1, \theta_2)\in\R^2$ (in radians), where $\theta_1=0$ is arbitrary and $\theta_2 = 0$ corresponds to the upright position. The state is then $(\theta_1, \theta_2, \dot\theta_1, \dot\theta_2)$.
We collect a set of output trajectories by initializing the system by hand close to $x_\a=(0,0,0,0) $ and $x_\b=(\pi/4,0,0,0)$, respectively. 
We measure both angles with encoders and corrupt these measurements with centered Gaussian noise of variance $\SI{0.001}{\radian\tothe{2}}$ independently on each dimension, in order to further complicate the experiment. 
The trajectories of $\SI{9.75}{\second}$ are sampled with $\dt=\SI{0.01}{\second}$.

\subsubsection{Choosing a Sensor Configuration}
We consider three possible settings: $y=(\theta_1, \theta_2)$, $y=\theta_1$ and $y=\theta_2$.
Our goal is to determine from the collected samples which of these configurations enable distinguishing the initial states, a necessary condition to provably reconstruct the whole state of the system.
\par 
For this, we run the kernel two-sample test of Theorem~\ref{thm:test} ten times in all three settings, and present the results in Table~\ref{table_qube_results}.
The test triggers every time when $y=(\theta_1, \theta_2)$ or $y=\theta_1$, indicating that the initial distributions are distinguishable with these outputs with high confidence.
In contrast, it never triggers when $y=\theta_2$; the data does not allow to distinguish between $x_\a$ and $x_\b$ when only measuring $\theta_2$.
\reworded{\reviewcomment{1.10}This is consistent with the deterministic analysis: computing the Lie observability matrix shows that the model is locally observable around the origin for $y=\theta_1$, but not for $y=\theta_2$.}

\begin{remark}
The test evaluates the distinguishability of the initial \emph{distributions} $\mu_\a$ and $\mu_\b$ whereas, in practice, one is usually concerned with the distinguishability of the initial \emph{states} $x_\a$ and $x_\b$; yet, $\mu_\a$ and $\mu_\b$ are not Diracs. 
This could lead to a false interpretation of the test result. 
For example, if $x_\a=x_\b$ but $\mu_\a \neq \mu_\b$ (\eg one is much more spread than the other), the output distributions may differ: the test would then trigger, although $x_\a$ and $x_\b$ are indistinguishable. 
To interpret the test result as distinguishability induced by the \emph{dynamics} and not as an artifact introduced by non-ideal initial distributions, the support of the latter should be as narrow as possible around $x_\a$ and $x_\b$.
This is theoretically sound due to the continuity of the \ac{MMD}: if the support of $\mu_\a$ and $\mu_\b$ is sufficiently small then $\mmd[\lawoutput{\mu_\a}{T},\lawoutput{\mu_\b}{T}]$ is close to $\mmd[\lawoutput{x_\a}{T},\lawoutput{x_\b}{T}]$.
\end{remark}

\reworded{In practice, one can ensure a small spread through careful reinitializations or by discarding outlying trajectories.
We demonstrate robustness to this issue by running our experiments multiple times and randomly selecting $n=m=40$ trajectories from each state.}

\subsubsection{Consequences for Observer Behavior}
Beyond choosing sensor configurations, the information provided by our test can also help in understanding the behavior of observers that fail to converge.
\reworded{\reviewcomment{7.8}%
When a system is initialized in two distributionally indistinguishable states, an observer takes as input the \emph{same} measurement distributions from both these states.
Therefore, the output of the observer itself has the same distribution; its estimates are statistically identical on the whole class of indistinguishability.
}
\par
We demonstrate this by implementing \iac{EKF} based on a model of the pendulum and taking $y$ as input. 
We run the \ac{EKF} on our data set in the two settings $y=\theta_1$ (distinguishable) and $y=\theta_2$ (indistinguishable).
The observer converges in the first case, but not in the second one due to lack of observability (Fig.\,\ref{fig_quanserqube_EKF}, top).
More interestingly, we estimate the evolution of the distribution of $\sin(\theta_1)$ as estimated by the~\ac{EKF} (Fig.\,\ref{fig_quanserqube_EKF}, bottom).
In the distinguishable case (Fig.\,\ref{fig:violin distinguishable}), the two output distributions of the~\ac{EKF} differ, whereas they are similar in the indistinguishable case (Fig.\,\ref{fig:violin indistinguishable}).
This gives a precise meaning to the idea that an observer can estimate the class of indistinguishability of the initial state, which our test can determine.

\begin{table}
	\centering
	\begin{tabular}{@{\extracolsep{\fill}}cccc}
		\toprule
		Measurement & $y = (\theta_1, \theta_2)$ & $y = \theta_1$ & $y = \theta_2$ \\
        \midrule
        $\sigma$ (\si{\radian}) & 1300 & 1300 & 2500 \\
        \midrule
        $\frac{\mmdb}{\thresh}$ & 1.40 & 2.50 & 0.34 \\
        \midrule
        Trigger (\% of trials) & $100$ & $100$ & $0$\\
		\bottomrule
	\end{tabular}
	\caption{Results for the Furuta pendulum, averaged over ten tests. Only the first two configurations are identified as distinguishable.}
	\label{table_qube_results}
\end{table}

\section{Conclusion and Outlook}

\reworded{\reviewcomment{1.1}}In this paper, we present a method to assess the \reworded{distinguishability} of stochastic systems from data without explicitly relying on a dynamics model; this is the first such method to the best of our knowledge.
\par
We start by extending distributional distinguishability to arbitrary nonlinear stochastic systems and investigate its relationship to deterministic observability, showing that both notions are equivalent for certain classes of systems which include linear systems.
We then introduce the \ac{MMD} and a finite-sample estimate thereof to quantify distinguishability of initial states from noisy output trajectories.
We leverage the \ac{MMD} for a statistical test of distinguishability, hereby interpreting its value as how much data is required to tell initial states apart with confidence.
We illustrate the proposed tools and notions in simulation and on hardware.
We reveal experimentally the relationship between noise and distinguishability, emphasizing the interest of a continuous metric, and showcase how such an approach could benefit observer design by allowing for a priori error analysis.
\par
The theoretical analysis of distributional observability that we propose can be further pushed, e.g., by characterizing other classes of measurement noise that satisfy Assumption~\ref{asmptn:injectivity measurement}, or \added{by leveraging assumptions on the control to allow inferring distinguishability for one controller from data collected with another.\reviewcomment{10.5}}

\section*{Acknowledgment}
We thank Sebastian Giedyk for his help with the experiments.

\ifarxiv
\bibliography{IEEEabrv,references.bib}
\else
\printbibliography
\fi

\end{document}